\documentclass[a4paper,11pt]{article}
\usepackage{wrapfig}
\usepackage[dvips]{graphics}
\usepackage{latexsym}
\usepackage{graphicx}
\usepackage{epsfig}
\usepackage[linesnumbered,ruled,vlined]{algorithm2e}
\usepackage{algorithmic}

\usepackage{amssymb}
\usepackage{amsfonts}
\usepackage{color}

\newcommand{\remove}[1]{}

\newcommand{\IR}{\mathbb{R}}

\newtheorem{theorem}{Theorem}
\newtheorem{lemma}{Lemma}

\newtheorem{definition}{Definition}

\newenvironment{proof}{\noindent {\bf Proof:\,\ }}{\hfill\mbox{\
$\Box$}\smallskip}
\title{Minimum Dominating Set for a Point Set in $\IR^2$}
\author{Ramesh K. Jallu\thanks{Indian Institute of Technology Guwahati,
India}
\and Prajwal R. Prasad \thanks{National Institute of Technology Karnataka, India}
\and Gautam K. Das \footnotemark[1]}
\date{}
\begin{document}
\maketitle
\begin{abstract}
In this article, we consider the problem of computing minimum dominating set for a given set $S$ 
of $n$ points in $\IR^2$. Here the objective is to find a minimum cardinality 
subset $S'$ of $S$ such that the union of the unit radius disks centered at the points in $S'$ covers 
all the points in $S$. We first propose a simple 4-factor and 3-factor approximation algorithms in $O(n^6 \log n)$ 
and $O(n^{11} \log n)$ time respectively improving time complexities by a factor of $O(n^2)$ and 
$O(n^4)$ respectively over the best known result available in the literature 
[M. De, G.K. Das, P. Carmi and S.C. Nandy, {\it Approximation 
algorithms for a variant of discrete piercing set problem for unit disk}, Int. J. of Comp. Geom. and Appl., 
to appear]. Finally, we propose a very important shifting lemma, which is of independent 
interest and using this lemma we propose a $\frac{5}{2}$-factor approximation algorithm and a PTAS for 
the minimum dominating set problem. 
\end{abstract}


\section{Introduction}
A minimum dominating set $S'$ for a set $S$ of $n$ points in $\IR^2$ is defined as follows: (i) $S' \subseteq S$ 
(ii) each point $s \in S$ is covered by at least one unit radius disk centered at a point in $S'$, and (iii) size of 
$S'$ is minimum. The {\it minimum dominating set} (MDS) problem for a point set $S$ of size $n$ in $\IR^2$ 
involves finding a minimum dominating set $S'$ for the set $S$. We call this problem as a geometric version of 
MDS problem. The MDS problem for a point set can be modeled  
as an MDS problem in unit disk graph (UDG) as follows: A unit disk graph $G = (V,E)$ for a set ${\cal U}$ of $n$ 
unit diameter disks in $\IR^2$ is the intersection graph of the family of disks in ${\cal U}$ i.e., the vertex 
set $V$ corresponds to the set ${\cal U}$ and two vertices are connected by an edge if the corresponding disks have 
common intersection. The minimum dominating set for the graph $G$ is a minimum size subset $V'$ of $V$ such 
that for each of the vertex $v \in V$ is either in $V'$ or adjacent to to a node in $V'$ in $G$. Several people have done 
research on MDS problem because of its wide applications such as wireless networking,  
facility location problem, to name a few. Our interest in this problem arose from the following reason: 
suppose in a city we have a set $S$ of $n$ important locations (houses, etc.); the objective is to provide 
some emergency services (ambulance, fire station, etc.) to each of the locations in $S$ so that 
each location is within a predefined distance of at least one service center. Note that positions of the 
emergency service centers are from the predefined set $S$ of locations only.

\subsection{Related Work}
The MDS problem can be viewed as a general set cover problem, but it is an NP-hard problem \cite{GJ79,J82} 
and not approximable within $c \log n$ for some constant $c$ unless P = NP \cite{RS97}. Therefore 
$O(\log n)$-factor approximation algorithm is possible for MDS problem by applying the algorithm for 
general set cover problem \cite{chvatal79}. Some exciting results for the geometric version of MDS problem 
are available in the literature. 

In the {\it discrete unit disk cover} (DUDC) problem, two sets $P$ and $Q$ of points in $\IR^2$ are given, the 
objective is to choose minimum number of unit disks $D'$ centered at the points in $Q$ such that 
the union of the disks in $D'$ covers all the points in $P$. Johnson \cite{J82} proved that 
the DUDC problem is NP-hard. Mustafa and Ray in 2010 \cite{MR10} proposed a $(1+\delta)$-approximation 
algorithm for $0 < \delta \leq 2$ (PTAS) for the DUDC problem using $\epsilon$-net based local improvement approach. 
The fastest algorithm is obtained by setting $\delta = 2$ for a 3-factor approximation algorithm, which runs 
in $O(m^{65}n$) time, where $m$ and $n$ are number of unit radius disks and number of points respectively \cite{DFLN12}. 
The high complexity of the PTAS leads to further research on constant factor 
approximation algorithms for the DUDC problem. A series of constant factor approximation algorithms for DUDC 
problem are available in the literature: 

\begin{itemize}
\item 108-approximation algorithm [C\u{a}linescu et al., 2004 \cite{CMWZ04}]
\item 72-approximation algorithm [Narayanappa and Voytechovsky, 2006 \cite{NV06}]
\item 38-approximation algorithm in O($m^{2}n^{4}$) time [Carmi et al., 2007 \cite{CKL07}]
\item 22-approximation algorithm in O($m^{2}n^{4}$) time [Claude et al., 2010 \cite{CDDDFLNS10}]
\item 18-approximation algorithm in O($mn+n\log n+m\log m$) time [Das et al., 2012 \cite{DFLN12}]
\item 15-approximation algorithm in O($m^{6}n$) time [Fraser and L\'{o}pez-Ortiz, 2012 \cite{FL12}]
\item $(9+\epsilon)$-approximation algorithm in $O(m^{3(1+\frac{6}{\epsilon})} n \log n)$ time 
[Acharyya et al., 2013 \cite{ABD13}]
\end{itemize}

The DUDC problem is a geometric version of MDS problem for $P=Q$. Therefore all results for the DUDC problem 
are applicable to MDS problem.

The geometric version of MDS problem is known to be NP-hard \cite{CCJ90}. 
Nieberg and Hurink \cite{NH06} proposed $(1+\epsilon)$-factor approximation algorithm 
for $0 < \epsilon \leq 1$. The fastest algorithm is obtained by setting $\epsilon = 1$ for a
$2$-approximation result, which runs in $O(n^{81})$ time \cite{DDCN13}, which is not 
practical even for $n = 2$. Another PTAS for dominating set of arbitrary size disk graph is available 
in the literature proposed by Gibson and Pirwani \cite{GP10}. The running time of this PTAS is 
$n^{O(\frac{1}{\epsilon^2})}$. 

Marathe et al. \cite{MBIRR95} proposed a 5-factor approximation algorithm for the MDS problem. 
Amb{\"u}hl et al. \cite{AEMN06} proposed 72-factor approximation algorithm for weighted dominating 
set (WDS) problem. In the WDS problem, each node has a positive weight and the objective is to find 
the minimum weight dominating set of the nodes in the graph. Huang et al. \cite{HGZW08}, Dai and Yu 
\cite{DY09}, and  Zou et al. \cite{ZWXLDWW11} improved the approximation factor for WDS problem to 
$6+\epsilon$, $5+\epsilon$, and $4+\epsilon$ respectively. First, they proposed $\gamma$-factor 
($\gamma = 6, 5, 4$ in \cite{HGZW08}, \cite{DY09}, and \cite{ZWXLDWW11} respectively) approximation 
algorithm for a subproblem and using the result of their corresponding sub-problems they proposed 
$(\gamma+\epsilon)$-factor approximation algorithms. The time complexity of their algorithms 
are $O(\alpha(n) \times \beta(n))$, where $O(\alpha(n))$ is the time complexity of the algorithm 
for the sub-problem and $O(\beta(n)) = O(n^{4 (\lceil\frac{84}{\epsilon} \rceil)^2})$ is the number 
of times the sub-problem needs to be invoked to solve the original problem. The $(\gamma + 1)$-factor 
approximation algorithm can be obtained by setting $\epsilon = 1$, but the time complexity becomes a
very high degree polynomial function in $n$. Carmi et al. \cite{CKL08} proposed a 5-factor approximation 
algorithm of the MDS problem for arbitrary size disk graph. Fonseca et al. \cite{FFSM12} proposed a 
$\frac{44}{9}$-factor approximation algorithm for the MDS problem in UDG which can be achieved in $O(n+m)$ 
time, when the input is a graph with $n$ vertices and $m$ edges, and in $O(n \log n)$ time, in the geometric 
version of the problem. The same set of authors also proposed a $\frac{43}{9}$-factor approximation algorithm 
for the MDS problem in UDG which runs in $O(n^2 m)$ time \cite{FFSM12-1}. Recently, De at al. \cite{DDCN13} 
considered the geometric version of MDS problem and proposed 12-factor, 4-factor, and 3-factor approximation 
algorithms with running time $O(n \log n)$, $O(n^8 \log n)$, and $O(n^{15}\log n)$ respectively. They also 
proposed a PTAS with high degree polynomial running time. 

\subsection{Our Contribution}
In this paper, we consider the geometric version of MDS problem and propose a series of constant factor 
approximation algorithms. We first propose 4-factor and 3-factor approximation algorithms with running time 
$O(n^6 \log n)$ and $O(n^{11}\log n)$ respectively improving the time complexities by a factor of $O(n^2)$ and 
$O(n^4)$ respectively over the best known result in the literature \cite{DDCN13}. Finally, we propose a new 
shifting strategy lemma. Using our shifting strategy lemma we propose $\frac{5}{2}$-factor and 
$(1+\frac{1}{k})^2$-factor (i.e., PTAS) approximation algorithms for the MDS problem. The running time of 
proposed $\frac{5}{2}$-factor and $(1+\frac{1}{k})^2$-factor approximation algorithms are $O(n^{20} \log n)$ 
and $n^{O(k)}$ respectively. Though the time complexity of the proposed PTAS is same as the PTAS proposed by 
De et al. \cite{DDCN13} in terms of $O$ notation, but the constant involved in our PTAS is smaller than the 
same in \cite{DDCN13}.

\section{4-Factor Approximation Algorithm for the MDS Problem}\label{4factor}
In this section, a set $S$ of $n$ points in $\mathbb{R}^2$ is given inside a rectangular region ${\cal R}$. The 
objective is to find an MDS for $S$. Here we propose a simple 4-factor approximation algorithm. The running time 
of our algorithm is $O(n^6 \log n)$, which is an improvement by a factor of $O(n^2)$ over the best known existing 
result \cite{DDCN13}. In order to obtain a 4-factor approximation algorithm, we consider a partition of 
${\cal R}$ into regular hexagons of side length $\frac{1}{2}$ (see Figure \ref{figure-2}(a)). We use {\it cell} 
to denote a regular hexagon of side length $\frac{1}{2}$.

\begin{lemma} \label{lemma-1x}
 All points inside a single cell can be covered by an unit radius disk centered at any point inside that cell.
\end{lemma}

\begin{proof}
The lemma follows from the fact that the distance between any two points inside a regular hexagon of side length 
$\frac{1}{2}$  is at most 1 (for demonstration see the Figure \ref{figure-2}(b)).
\end{proof}

\begin{figure}[!ht]
\begin{center} 
\includegraphics[]{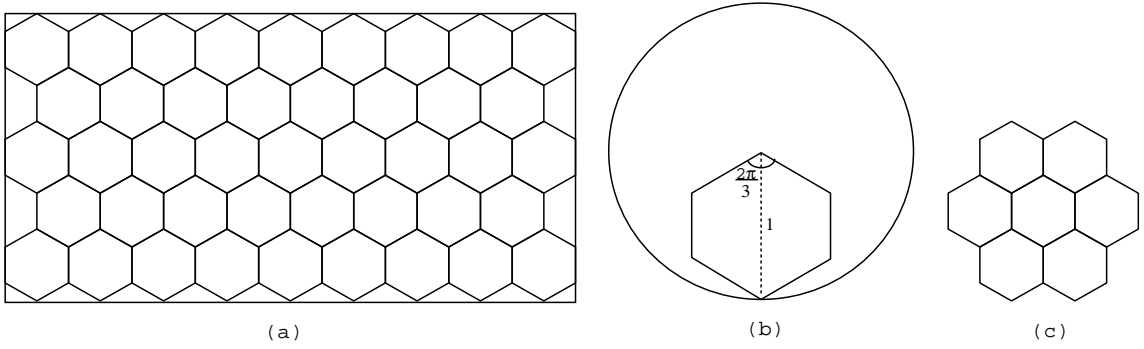}\\
\caption{(a) Regular hexagonal partition (b) single regular hexagon of side length $\frac{1}{2}$ contained 
in an unit radius disk, and (c) a septa-hexagon}
\label{figure-2}
\end{center}
\vspace{-0.2in}
\end{figure}

\begin{definition}
A {\it septa-hexagon} is a combination of 7 adjacent cells such that one cell is inscribed by six other 
cells as shown in Figure \ref{figure-2}(c). 

For a point set $U$, we use $\Delta(U)$ to denote the set of unit radius disks centered at the points in $U$.

Let $U_1$ and $U_2$ be two point sets such that $U_1 \subseteq U_2$. We use $\chi(U_1, U_2)$ to denote the set 
of points such that $\chi(U_1, U_2) \subseteq U_2$ and an unit radius disk centered at any point in 
$\chi(U_1, U_2)$ covers at least one point of $U_1$.
\end{definition}

\subsection{Algorithm overview}
Let us consider a septa-hexagon ${\cal C}$. Recall that ${\cal C}$ is a combination of 7 cells (regular 
hexagon of side length $\frac{1}{2}$). Let $S_1 = S \cap {\cal C}$ and $S_2 = \chi(S_1, S)$. For the 
4-factor approximation algorithm, we first find minimum size subset $S' \subseteq S_2$ such that 
$S_1 \subseteq \bigcup_{d \in \Delta(S')}d$. Call this problem as {\it single septa-hexagon MDS} problem. 
Using the optimum (minimum size) solution of single septa-hexagon MDS problem, we present our main 4-factor 
approximation algorithm. The Lemma \ref{lemma-2x} gives an important feature to design optimum algorithm for
single septa-hexagon MDS problem.

\begin{lemma} \label{lemma-2x}
 If $OPT_{\cal C}$is a minimum cardinality subset of $S_2$ such that 
$S_1 \subseteq \bigcup_{d \in \Delta(OPT_{\cal C})}d$, then $|OPT_{\cal C}| \leq 7$.
\end{lemma}

\begin{proof}
The septa-hexagon ${\cal C}$ has at most 7 non-empty cells. From Lemma \ref{lemma-1x}, we know that 
an unit radius disk centered at a point in a cell covers all points in that cell. Therefore one point 
from each of the non-empty cells is sufficient to cover all the points in ${\cal C}$. Thus the 
Lemma follows. 
\end{proof}

\begin{algorithm}[!ht]
\caption{Algorithm\_4\_Factor($S, {\cal C}, n$)}
\begin{algorithmic}[1]
\STATE {\bf Input:} A set $S$ of $n$ points and a septa-hexagon ${\cal C}$

\STATE {\bf Output:} A set $S' (\subseteq S)$ such that $(S \cap {\cal C}) \subseteq \bigcup_{d \in \Delta(S')}d$. 

\STATE $S' \leftarrow \emptyset$
\IF{($S \cap {\cal C} \neq \emptyset$)}
  \STATE Choose one arbitrary point from each non-empty cell of ${\cal C}$ and add to $S'$.
  \STATE $m \leftarrow |S'|$ /* $m$ is at most 7 */
  \STATE  Let $S_1 = S \cap {\cal C}$ and $S_2 = \chi(S_1, S)$. 
  \FOR{($i = m-1, m-2, \ldots, 1$)}
    \IF{($i=6$)}
      \FOR {(Each possible combination of 5 points $X = \{p_1, p_2, \ldots, p_5\}$ of $S_2$)}
	  \STATE Find $Y \subseteq S_1$ such that no point in $Y$ is covered by $\bigcup_{d \in \Delta(X)}d$.
         
          \STATE Compute the farthest point Voronoi diagram of $Y$ \cite{BCKO08}
          \STATE Find a point $p$ (if any) from $S_2\setminus X$ (using planar point location algorithm \cite{PS09}) 
	    such that the farthest point in $Y$ from $p$ is less than or equal to 1. If such $p$ exists, then 
	    set $S' \leftarrow X \cup \{p\}$ and exit {\bf for} loop.
      \ENDFOR
    \ELSE
      \FOR {(Each possible combination of $i$ points $X = \{p_1, p_2, \ldots, p_i\}$ of $S_2$)}
	\IF{($S_1 \subseteq \bigcup_{d \in \Delta(X)}d$)}
	  \STATE Set $S' \leftarrow X$ and exit from {\bf for} loop 
	\ENDIF
      \ENDFOR
    \ENDIF
  \ENDFOR
\ENDIF
\STATE Return $S'$
\end{algorithmic}
\label{algo-4factor}
\end{algorithm}

\begin{lemma} \label{lemma-3x}
For a given set $S$ of $n$ points and a septa-hexagon ${\cal C}$, the Algorithm \ref{algo-4factor} computes an 
MDS for $S \cap {\cal C}$ using the points of $S$ in $O(n^6 \log n)$ time.
\end{lemma}

\begin{proof}
 The optimality of the Algorithm \ref{algo-4factor} follows from the fact that Algorithm \ref{algo-4factor} 
 considers all possible set of sizes $0, 1, \ldots, 7$ (see Lemma \ref{lemma-2x}) as its solution and 
 reports minimum size solution.
 
 The line number 7 of the algorithm can be computed in $O(n \log n)$ time as follows: (i) computation of the set 
 $S_1$ takes $O(n)$ time, (ii) computation of $S_2$ can be done in $O(n \log n)$ time using nearest 
 point Voronoi diagram of $S_1$ in $O(n \log n)$ time and for each point $p \in S$ apply planar point 
 location algorithm to find the nearest point in $S_1$ in $O(\log n)$ time.
 
 The running time of the {\bf else} part in the line number 15 of the algorithm is at most $O(n^6)$ time. 
 The worst case running time of the algorithm comes from line numbers 9-14. The complexity of line numbers 11-13 
 is $O(n \log n)$ time. Therefore the running time of the line numbers 9-14 is $O(n^6 \log n)$ time. 
 Thus the overall worst case running time of the proposed Algorithm \ref{algo-4factor} is $O(n^6 \log n)$.
\end{proof}

Let us consider a septa-hexagonal partition of ${\cal R}$ such that no point of $S$ is on the boundary of 
any septa-hexagon and a 4 coloring scheme of it (see Figure \ref{fig:fig10}). Consider an unicolor 
septa-hexagon of color A (say). Its adjacent septa-hexagons are assigned colors B, C and D (say) such 
that opposite septa-hexagons are assigned the same color (see Figure \ref{fig:fig10}).  

\begin{figure}[ht]
\begin{center}
\includegraphics[height=3in]{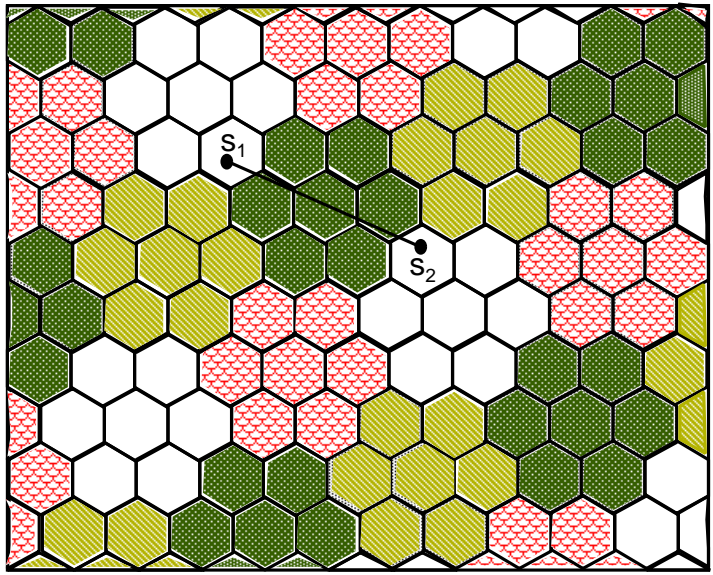}
\caption{A septa-hexagonal partition and 4-coloring scheme}
\label{fig:fig10}
\end{center}
\end{figure}

\vspace{-0.1in}
\begin{lemma} \label{lemma-4x}
 If ${\cal C}'$ and ${\cal C}''$ are two same colored septa-hexagons, then 
 $({\cal C}' \cup {\cal C}'') \cap S \cap d = \emptyset$ for any unit radius disk $d$.
\end{lemma}

\begin{proof}
According to the 4-coloring scheme, size of the septa-hexagons, and no point of $S$ is on the boundary 
of ${\cal C}'$ and ${\cal C}''$ the minimum distance between two points $s_1 \in {\cal C}' \cap S$ 
and $s_2 \in {\cal C}''\cap S$) is greater than 2 (see Figure \ref{fig:fig10}). Thus the lemma follows.
\end{proof}

\begin{theorem} \label{theorem-1y}
The 4-coloring scheme gives a 4-factor approximation algorithm for the MDS problem in $O(n^{6}\log n)$ time, 
where $n$ is the input size.
\end{theorem}

\begin{proof}
Let $N_1, N_2, N_3$, and $N_4$ be the sets of septa-hexagons of colors $A, B, C$, and $D$ respectively. 
Let $S_1^i = S \cap \bigcup_{{\cal C} \in N_i} {\cal C}$ and $S_2^i = \chi(S_1^i, S)$ 
for $1 \leq i \leq 4$. By Lemma 
\ref{lemma-4x}, the pair ($S_1^i, S_2^i$) can be partitioned into $|N_i|$ pairs ($S_{1j}^i, S_{2j}^i$) 
such that for each pair Algorithm \ref{algo-4factor} is applicable for solving the covering problem optimally 
to cover $S_1^i$ using $S_2^i$, where $1 \leq j \leq |N_i|$. Let $N_i'$ be the optimum solution for the set  
$S_1^i$ ($1 \leq i \leq 4$) using the Algorithm \ref{algo-4factor}. If $OPT$ is the optimum solution for the 
set $S$, then $|N_i'| \leq |OPT|$. Therefore $\Sigma_{i=1}^4 |N_i'| \leq 4 \times |OPT|$. Thus the 
approximation factor of the algorithm follows.

The time complexity result of the theorem follows from Lemma \ref{lemma-3x} and the fact that each point in 
$S$ can participate in the Algorithm \ref{algo-4factor} at most constant number of times. 
\end{proof}

\section{3-Factor Approximation Algorithm for the MDS Problem}\label{3factor}
Given a set $S$ of $n$ points in a rectangular region ${\cal R}$, we wish to find an MDS for $S$. 
Here we present a 3-factor approximation algorithm in $O(n^{11} \log n)$ time for the MDS 
problem, which is an improvement by a factor of $O(n^4)$ over the best known result available in 
the literature \cite{DDCN13}. 

\begin{definition}
A {\it super-cell} is a combination of 15 regular hexagons of side length $\frac{1}{2}$ arranged 
in three consecutive rows as shown in Figure \ref{fig:fig14}. 
\end{definition}

\begin{figure}[ht]
\begin{center}
\includegraphics[]{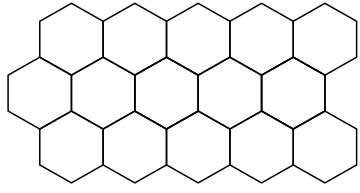}
\caption{An example of a super-cell}
\label{fig:fig14}
\end{center}
\end{figure}

\subsection{Algorithm overview}
Let us consider a {\it super-cell} ${\cal D}$. Let $S_1 = S \cap {\cal D}$ and $S_2 = \chi(S_1, S)$. 
In order to obtain 3-factor approximation algorithm for the MDS problem, we first find a minimum 
size subset $S' \subseteq S_2$ such that $S_1 \subseteq \bigcup_{d \in \Delta(S')}d$. Call this problem 
as a {\it single super-cell MDS} problem. Using the optimum solution of single super-cell MDS problem, we 
present our main 3-factor approximation algorithm. 

\begin{lemma} \label{lemma-5x}
If $OPT_{\cal D}$ is the minimum cardinality subset of $S_2$ such that 
$S_1 \subseteq \bigcup_{d \in \Delta(OPT_{\cal D})} d$, then $|OPT_{\cal D}| \leq 15$.
\end{lemma}

\begin{proof}
The lemma follows from the Lemma \ref{lemma-1x} and the fact that the super-cell ${\cal D}$ 
has at most 15 non-empty cells.
\end{proof}

We decompose a super-cell ${\cal D}$ into 3 regions namely $G_{\cal D}^1, G_{\cal D}^2$, and $G_{\cal D}^3$ 
(see Figure \ref{fig:fig16},  where $G_{\cal D}^1, G_{\cal D}^2$, and $G_{\cal D}^3$ correspond to unshaded, 
light shaded, and dark shaded regions respectively). 

\begin{figure}[ht]
\begin{center}
\includegraphics[]{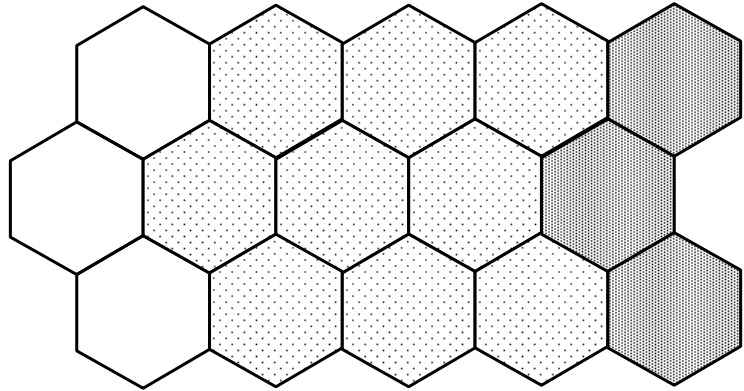}
\caption{Decomposition of a super-cell}
\label{fig:fig16}
\end{center}
\end{figure}

\begin{lemma} \label{lemma-6x}
For any unit radius disk $d$ and a super-cell ${\cal D}$, $(G_{\cal D}^1 \cup G_{\cal D}^3) \cap d = \emptyset$. 
\end{lemma}

\begin{proof}
The lemma follows from the fact that if $s$ and $t$ are two arbitrary points of $G_{\cal D}^1$ and $G_{\cal D}^3$ 
respectively, then the Euclidean distance between $s$ and $t$ is greater than 2. 
\end{proof}

Let $S_1 = S \cap {\cal D}$ and $S_2 = \chi(S_1, S)$, where ${\cal D}$ is a super-cell. Our objective 
is to find a minimum cardinality set $S' (\subseteq S_2)$ such that $S_1 \subseteq \bigcup_{d \in \Delta(S')}d$.

Let $S_1^1 =  S_1 \cap G_{\cal D}^1$, $S_1^2 =  S_1 \cap G_{\cal D}^2$, and $S_1^3 =  S_1 \cap G_{\cal D}^3$. 
A point on a boundary can be assigned to any set associated with that boundary. Let $S_2^1 = \chi(S_1^1, S_2)$, 
$S_2^2 = \chi(S_1^2, S_2)$, and $S_2^3 = \chi(S_1^3, S_2)$. The Lemma \ref{lemma-6x} says that 
$S_2^1 \cap S_2^3 = \emptyset$.

\begin{algorithm}
\caption{Algorithm\_3\_Factor($S, {\cal D}, n$)}
\begin{algorithmic}[1]
\STATE {\bf Input:} A set $S$ of $n$ points and a super-cell ${\cal D}$

\STATE {\bf Output:} A set $S' (\subseteq S)$ such that $(S \cap {\cal D}) \subseteq \bigcup_{d \in \Delta(S')}d$

\STATE $S' \leftarrow S$.
\STATE Find the sets $S_1^1, S_1^2,S_1^3, S_2^1, S_2^2$,and $S_2^3$ as defined above.

\FOR {(Each possible combination $X = \{p_1, p_2, \ldots, p_j\}$ of $j (0 \leq j \leq 9)$ points in $S_2^2$)}
  \IF {($S_1^2 \subseteq \bigcup_{d \in \Delta(X)}d$)}
	\STATE Let $U$ and $V$ be the subsets of $S_1^1$ and $S_1^3$ respectively such that no point in 
	$U \cup V$ is covered by $\bigcup_{d \in \Delta(X)}d$.
	
	\STATE Let $Y$ be the minimum size subset of $S_2^1$ such that $U \subseteq \bigcup_{d \in \Delta(Y)}d$.
	  
	\STATE Let $Z$ be the minimum size subset of $S_2^3$ such that $V \subseteq \bigcup_{d \in \Delta(Z)}d$.
	
	\IF{($|S'| > |X| + |Y| + |Z|$)}
	    \STATE Set $S' \leftarrow X \cup Y \cup Z$
	\ENDIF
  \ENDIF
\ENDFOR
\STATE Return $S'$
\end{algorithmic}
\label{algo-3factor}
\end{algorithm}

\begin{lemma} \label{lemma-7x}
For a given set $S$ of $n$ points and a super-cell ${\cal D}$, the Algorithm \ref{algo-3factor} computes an 
MDS for $S \cap {\cal D}$ using the points of $S$ in $O(n^{11} \log n)$ time.
\end{lemma}

\begin{proof}
 In the case of selecting 3 points in $S_2^1$ in line number 8 of the algorithm, we can choose one point from 
 each of the non-empty cells  of $G_{\cal D}^1$. Therefore, the worst case of line number 8 appears for the 
 case of choosing all possible combinations of two points in $S_2^1$. This can be done in $O(n^2 \log n)$ using 
 the technique of the Algorithm \ref{algo-4factor} (line numbers 12-13). Similar analysis is applicable to line 
 number 9. Line numbers 6-7 and 10-12 can be implemented in $O(n)$ time.
 
 The worst case running time of the algorithm depends on the {\bf for} loop in the line number 5. In this 
 {\bf for} loop, we are choosing all possible 9 points from a set of $n$ points in worst case. Therefore the
 time complexity of the Algorithm \ref{algo-3factor} is $O(n^{11} \log n)$.  
 
 The optimality of the algorithm follows from the Lemma \ref{lemma-6x} and fact that Algorithm \ref{algo-3factor} 
 considers all possible combinations as its solution and reports minimum size solution.
 
 Note that Algorithm \ref{algo-3factor} checks {\bf if} condition in line number 6 because of the definition 
 of $S_2^1, S_2^2$, and $S_2^3$. 
\end{proof}

Let us consider a super-cell partition of ${\cal R}$ such that no point of $S$ lies on the boundary and a 
3-coloring scheme (see Figure \ref{fig:fig15}). Consider an unicolor super-cell which has been assigned 
color A (say). Its adjacent super-cells are assigned colors B, and C alternately (see Figure~\ref{fig:fig15}). 

\begin{figure}[ht]
\begin{center}
\includegraphics[height=4cm,width=9cm]{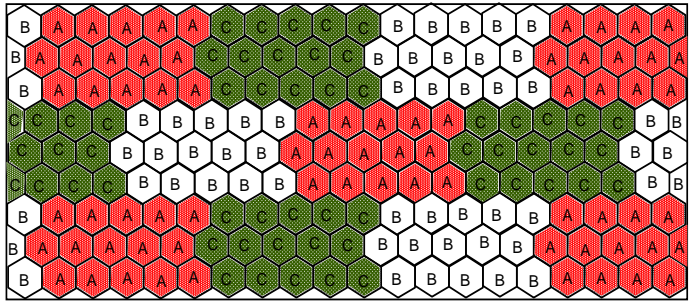}
\caption{A super-cell partition and 3-coloring scheme}
\label{fig:fig15}
\end{center}
\end{figure}

\begin{lemma} \label{lemma-8x}
 If ${\cal D}'$ and ${\cal D}''$ are two same colored super-cells, then 
 $({\cal D}' \cup {\cal D}'') \cap S \cap d = \emptyset$ for any unit radius disk $d$.
\end{lemma}

\begin{proof}
The lemma follows from the following facts: (i) size of the super-cells ${\cal D}'$ and ${\cal D}''$ 
(ii) no point of $S$ on the boundary of ${\cal D}'$ and ${\cal D}''$, and (iii) the 3-coloring scheme.  
\end{proof}

\begin{theorem} \label{theorem-2y}
 The 3-coloring scheme gives a 3-factor approximation algorithm for the MDS problem in 
 $O(n^{11} \log n)$ time, where $n$ is the input size.
\end{theorem}

\begin{proof}
The follows by the similar argument of Theorem \ref{theorem-1y}.   
\end{proof}

\section{Shifting Strategy and its Application to the MDS Problem} \label{ShiftingStrategy}
In this section, we first propose a shifting strategy for the MDS problem, which is 
a generalization of the shifting strategy proposed by Hochbaum and Maass \cite{HM85}. Next we propose 
$\frac{5}{2}$-factor approximation algorithm and a PTAS algorithm for MDS problem using our shifting strategy.  

\subsection{The Shifting Strategy} \label{shifting-strategy}
Our shifting strategy is very similar to the shifting strategy in \cite{HM85}. We include a brief 
discussion here for completeness. Let a set $S$ of $n$ points be distributed inside an axis aligned 
rectangular region ${\cal R}$. Our objective is to find an MDS for $S$. 

\begin{definition}
A {\it monotone chain} $c$ with respect to line $L$ is a chain of line segments 
such that any line perpendicular to $L$ intersect it only once. We define the distance between two 
monotone chains $c'$ and $c''$ as the minimum Euclidean distance between any two points $p'$ and $p''$ 
on the chains $c'$ and $c''$ respectively. A {\it monotone strip} denoted by $M_s$ and is defined by the 
area bounded by any two monotone chains $c'$ and $c''$ such that the area is left closed and 
right open.
\end{definition}

Consider a set $c_1, c_2, \ldots, c_r$ of $r$ monotone chains with respect to the line parallel to $y$-axis 
from left to right dividing the region ${\cal R}$ such that distance between each pair of monotone chains is 
at least $D (>0)$, where $c_1$ and $c_r$ are the left and right boundary of ${\cal R}$ respectively 
(see Figure \ref{shifting}). Let ${\cal A}$ be an $\alpha$-factor approximation algorithm, which provides a 
solution of any $\ell$ consecutive monotone strips for the MDS problem. 

\begin{figure*}[!ht]
\begin{center}
\includegraphics[height=2.2in]{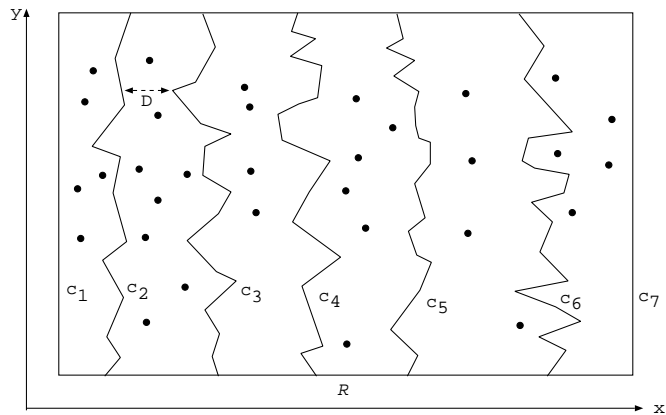}
\end{center} 
\caption{Demonstration of shifting strategy}
\label{shifting}
\end{figure*}

\begin{theorem} \label{shift-factor}
We can design an $\alpha(1+\frac{1}{\ell})$-factor approximation algorithm for finding an MDS for $S$.
\end{theorem}

\begin{proof}
The algorithm is exactly same as the algorithm proposed by Hochbaum and Maass \cite{HM85}. The approximation factor 
follows from exactly the same argument proved in the shifting lemma \cite{HM85}. 
\end{proof}

\subsection{$\frac{5}{2}$-Factor Approximation Algorithm for the MDS Problem}
Here we propose a $\frac{5}{2}$-factor approximation algorithm for MDS problem for a given 
set $S$ of $n$ points in $\IR^2$ using shifting strategy discussed in Subsection \ref{shifting-strategy}. 

\begin{definition}
 A {\it duper-cell} is a combination of 30 cells (regular hexagon of side length $\frac{1}{2}$) as shown in 
 Figure \ref{figure-8}. A duper-cell ${\cal E}$ generates four monotone chains with respect to 
 vertical and horizontal lines along its boundary. See Figure \ref{figure-8}, where $uv, vw, wx$, and $xu$ 
 are the monotone chains. We rename them as {\bf left, bottom, right}, and {\bf top} monotone chains.
\end{definition}

The basic idea is as follows: first optimally solve the subproblem {\it duper-cell} i.e., 
find an MDS for the set $S \cap {\cal E}$, where ${\cal E}$ is a duper-cell and then apply shifting 
strategy in both horizontal and vertical directions separately. The Lemma \ref{lemma-1x} leads 
to restriction on the size of the MDS, which is at most 30. Therefore an easy optimum solution 
for MDS can be obtained in $O(n^{30})$ time. Here we propose a different technique for the MDS problem 
leading to lower time complexity as follows:

\begin{figure*}[!ht]
\begin{center}
\includegraphics[height=1.5in]{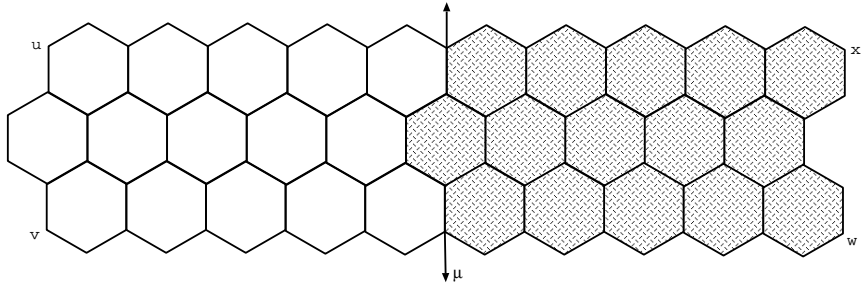}
\end{center} 
\caption{Demonstration of $\frac{5}{2}$-factor approximation algorithm}
\label{figure-8}
\end{figure*}

We divide the duper-cell ${\cal E}$ into 2 groups  unshaded region ($U_R$) and shaded region ($S_R$) as shown in 
Figure \ref{figure-8}. Let $\mu$ be the common boundary of the regions and two extended lines 
(see Figure \ref{figure-8}). Let $Q_1$ and $Q_2$ be two sets of points in the left (resp. right) of $\mu$ such 
that each disk in $\Delta(Q_1)$ and $\Delta(Q_2)$ intersects $\mu$. 

\begin{algorithm}[!ht]
\caption{MDS\_for\_duper-cell($S, {\cal E}, n$)}
\begin{algorithmic}[1]
\STATE {\bf Input:} A set $S$ of $n$ points and a duper-cell ${\cal E}$.
\STATE {\bf Output:} A set $S' (\subseteq S)$ for an MDS of $S \cap {\cal E}$.

\STATE Find $Q_1$ and $Q_2$ as described above. 
\STATE Let $S_2^L$ and $S_2^R$ be the set of points in $S\setminus (Q_1 \cup Q_2)$ such that each 
disk in $\Delta(S_2^L)$ and $\Delta(S_2^R)$ covers at least one point in $S \cap U_R$ and 
$S \cap S_R$ respectively. 

\STATE $S' \leftarrow \emptyset, X \leftarrow \emptyset$
\FOR {($i=0, 1, \ldots, 9$)}
  \STATE choose all possible $i$ disks in $\Delta(Q_1)$ (resp. $\Delta(Q_2)$) and for each combination 
  of $i$ disks find $S_1^L$ and $S_1^R$ such that $S_1^L \subseteq (S \cap U_R)$ and uncovered by that $i$ disks, 
  and $S_1^R \subseteq (S \cap S_R)$ and uncovered by that $i$ disks.
  
  \STATE Call Algorithm \ref{algo-3factor} for finding an MDS for the sets $S_1^L$ and $S_1^R$ separately.

\ENDFOR
\STATE Return $S'$
\end{algorithmic}
\label{algo-shifting}
\end{algorithm}

\begin{lemma} \label{lemma-9x}
An MDS for the set of points inside a duper-cell ${\cal E}$ can be computed optimally 
in $O(n^{20} \log n)$ time, where $n$ is the input size.
\end{lemma}

\begin{proof}
 The time complexity of line number 8 of the Algorithm \ref{algo-shifting} is $O(n^{11} \log n)$ 
 (see Lemma \ref{lemma-7x}). The line number 8 executes at most $O(n^9)$ time by the {\bf for} loop 
 in line number 6. Therefore the time complexity of the lemma follows.
 
 In the {\bf for} loop (line number 6 of the algorithm), we considered 
 all possible $i$ ($0 \leq i \leq 9$) disks in $\Delta(Q_1)$ and $\Delta(Q_2)$ separately. Since 
 the number of cells that can intersect with such $i$ disks is at most 9, therefore the range of $i$ is correct. 
 For each combination of $i$ disks, we considered all possible combinations to solve the problem 
 for $S_1^L$ and $S_1^R$ separately. Therefore the correctness of the algorithm follows. 
\end{proof}

\begin{theorem} \label{theorem-3y}
The shifting strategy discussed in Subsection \ref{shifting-strategy} gives a $\frac{5}{2}$-factor 
approximation algorithm, which runs in $O(n^{20} \log n)$ time for the MDS problem, where $n$ is the 
input size.
\end{theorem}

\begin{proof}
The distance between the monotone chains {\bf left} and {\bf right} of ${\cal E}$ is greater than 8,  
the distance between the monotone chains {\bf bottom} and {\bf top} is 2, and the diameter $(D)$ of the 
disks is 2. Now, if we apply shifting 
strategy in horizontal and vertical directions separately, then we get $(1+\frac{1}{4})(1+\frac{1}{1})$-factor 
i.e. $\frac{5}{2}$-factor approximation algorithm in $O(n^{20} \log n)$ time (see Lemma \ref{lemma-9x}) 
for the MDS problem. 
\end{proof}

\subsection{A PTAS for MDS Problem}
In this section, we present a $(1+\frac{1}{k})^2$-factor approximation algorithm in $n^{O(k)}$ time for 
a positive integer $k$. Suppose a set $S$ of $n$ points within a rectangular region ${\cal R}$ is given. 
Consider a partition of ${\cal R}$ into regular hexagonal cells of side length $\frac{1}{2}$. The idea 
of our algorithm is to solve the MDS problem optimally for the points inside regular hexagons (say ${\cal F}$) 
such that the distance between {\bf left} and {\bf right} (resp. {\bf bottom} and {\bf top}) monotone chains 
is $2k$ (see Figure \ref{figure-9}) and using our proposed shifting strategy carefully 
(see Subsection \ref{shifting-strategy}). 

\begin{figure*}[!ht]
\begin{center}
\includegraphics[height=2in]{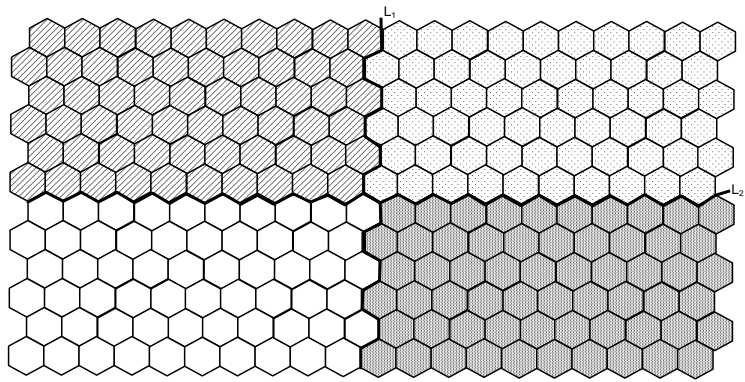}
\end{center} 
\caption{Demonstration of PTAS}
\label{figure-9}
\end{figure*}

To solve the MDS problem in $S \cap {\cal F}$ we further decompose ${\cal F}$ into four parts using 
the monotone chains $L_1$ and $L_2$ as shown in Figure \ref{figure-9}. The number of disks in the 
optimum solution intersecting the chain $L_1$ with centers {\bf left} (resp. {\bf right}) side of 
$L_1$ is at most $\lceil 3 \times 3 \times \frac{2k}{2} \rceil$ which is less than $10k$ and the 
number of disks in the optimum solution intersecting the chain $L_2$ with centers {\bf bottom} 
(resp. {\bf top}) side of $L_2$ is at most $\lceil 5 \times 3 \times \frac{2k}{4} \rceil$ which 
is less than $8k$. Next we apply recursive procedure to solve four independent sub-problems of 
size $k \times k$. If $T(n,2k)$ is the running time of the recursive algorithm for the MDS problem 
for $S \cap {\cal F}$, then using the technique of \cite{DDCN13} we have the following recurrence 
relation: $T(n, 2k) = 4 \times T(n, k) \times n ^{10k + 8k}$, which leads to the following theorem. 

\begin{theorem}\label{theorem-4y}
 For a given set $S$ of $n$ points in $\IR^2$, the proposed algorithm produces an MDS of  
 $S$ in $n^{O(k)}$ time, whose size is at most $(1+\frac{1}{k})^2 \times |OPT|$, where $k$ is a 
 positive integer and $OPT$ is the optimum solution. 
\end{theorem}

\section{Conclusion}
In this paper, we proposed a series of constant factor approximation algorithms for the MDS 
problem for a given set $S$ of $n$ points. Here we used hexagonal partition very carefully.
We first presented a simple 4-factor and 3-factor approximation algorithms in $O(n^6 \log n)$ 
and $O(n^{11} \log n)$ time respectively, which improved the time complexities of best known 
result by a factor of $O(n^2)$ and $O(n^4)$ respectively \cite{DDCN13}. Finally, we proposed a 
very important shifting lemma and using this lemma we presented a $\frac{5}{2}$-factor approximation 
algorithm and a PTAS for the MDS problem. Though the complexity of the proposed PTAS is same as that 
of the PTAS proposed by De et al. \cite{DDCN13} in terms of $O$ notation, but the constant involved 
in our PTAS is smaller than the same in \cite{DDCN13}.

\bibliographystyle{abbrv}

\end{document}